\documentclass[reqno,openright,dvips,11pt,myheadings,twoside]{amsart}

\usepackage{ifthen}

\usepackage[top=1.4in, bottom=1.4in, left=1.75in, right=1.75in]{geometry} 

\newcounter{arXivFormat} 

\newcommand{\IfarXivElse}[2]{
    \ifthenelse{\value{arXivFormat}=1}
        {#1}{#2}
    }

\newcounter{dtlForSubmission} 
\newcounter{dtlMarginComments} 
\newcounter{dtlSomeDetail} 
\newcounter{dtlFullDetails} 

\newcounter{DetailLevel} 

\newcommand{\DetailMarginNote}[1]{
    \ifthenelse{\value{DetailLevel}=\value{dtlMarginComments} \or \value{DetailLevel}>\value{dtlMarginComments}}
        {{\small #1}}{}
    }

\newcommand{\DetailSome}[1]{
    \ifthenelse{\value{DetailLevel}=\value{dtlSomeDetail} \or \value{DetailLevel}>\value{dtlSomeDetail}}
        {{\small \textbf{Detailed compile only}: #1}}{}
    }

\newcommand{\DetailFull}[1]{
    \ifthenelse{\value{DetailLevel}=\value{dtlFullDetails} \or \value{DetailLevel}>\value{dtlFullDetails}}
        {{\small \textbf{Detailed compile only}: #1}}{}
    }

\newcommand{\NotDetailSome}[1]{
    \ifthenelse{\value{DetailLevel}=\value{dtlSomeDetail} \or \value{DetailLevel}>\value{dtlSomeDetail}}
        {}{#1}
    }

\newcommand{\NotDetailFull}[1]{
    \ifthenelse{\value{DetailLevel}=\value{dtlFullDetails} \or \value{DetailLevel}>\value{dtlFullDetails}}
        {}{#1}
    }

\newcommand{\DetailSomeElse}[2]{
    \ifthenelse{\value{DetailLevel}=\value{dtlSomeDetail} \or \value{DetailLevel}>\value{dtlSomeDetail}}
        {{\small \textbf{Detailed compile only}: #1}}{#2}
    }

\newcommand{\DetailFullElse}[2]{
    \ifthenelse{\value{DetailLevel}=\value{dtlFullDetails} \or \value{DetailLevel}>\value{dtlFullDetails}}
        {{\small \textbf{Detailed compile only}: #1}}{#2}
    }

\newcommand{\DetailSomeInline}[1]{
    \ifthenelse{\value{DetailLevel}=\value{dtlSomeDetail} \or \value{DetailLevel}>\value{dtlSomeDetail}}
        {{\small #1}}{}
    }

\newcommand{\DetailFullInline}[1]{
    \ifthenelse{\value{DetailLevel}=\value{dtlFullDetails} \or \value{DetailLevel}>\value{dtlFullDetails}}
        {{\small #1}}{}
    }

\newcommand{\DetailSomeElseInline}[2]{
    \ifthenelse{\value{DetailLevel}=\value{dtlSomeDetail} \or \value{DetailLevel}>\value{dtlSomeDetail}}
        {{\small #1}}{#2}
    }

\newcommand{\DetailFullElseInline}[2]{
    \ifthenelse{\value{DetailLevel}=\value{dtlFullD	etails} \or \value{DetailLevel}>\value{dtlFullDetails}}
        {{\small #1}}{#2}
    }

\newcommand{\ExplainDetailLevel}{
    Detail level is
    \ifthenelse{\value{DetailLevel}=\value{dtlForSubmission}}
        {0: for submission}
        {\ifthenelse{\value{DetailLevel}=\value{dtlMarginComments}}
            {1: as for submission but with margin comments}
            {\ifthenelse{\value{DetailLevel}=\value{dtlSomeDetail}}
                {2: some proofs not intended for submission}
               {\ifthenelse{\value{DetailLevel}=\value{dtlFullDetails}}
                   {3: full details}
                   {invalid}
                }
            }
        }
    }

\newtheorem{theorem}{Theorem}[section]
\newtheorem{lemma}[theorem]{Lemma}

\newtheorem{cor}[theorem]{Corollary}

\theoremstyle{definition}

\numberwithin{equation}{section}

\newcommand{\abs}[1]{\left\vert#1\right\vert}

\newcommand{\innp}[1]{\ensuremath{\left< #1 \right>}}

\newcommand{\BoldTau}
    {\mbox{\boldmath $\tau$}}

\newcommand{\skipline}{\vspace{11pt}}

\newcommand{\BB}[1]{\ensuremath{\mathbb{#1}}}
\newcommand{\R}{\ensuremath{\BB{R}}} %
\newcommand{\iny}{\ensuremath{\infty}}
\newcommand{\grad}{\ensuremath{\nabla}}
\DeclareMathOperator{\dv}{div} %
\DeclareMathOperator{\curl}{curl} %
\DeclareMathOperator{\BV}{BV} %

\newcommand{\prt}{\ensuremath{\partial}}
\newcommand{\brac}[1]{\ensuremath{\left[ #1 \right]}}

\newcommand{\set}[1]{\ensuremath{\left\{ #1 \right\}}}

\newcommand{\norm}[1]{\ensuremath{\left\Vert #1 \right\Vert}}
\newcommand{\smallnorm}[1]{\ensuremath{\Vert #1 \Vert}}
\newcommand{\refS}[1]{Section~\ref{S:#1}}
\newcommand{\refT}[1]{Theorem~\ref{T:#1}}
\newcommand{\refL}[1]{Lemma~\ref{L:#1}}

\newcommand{\refC}[1]{Corollary~\ref{C:#1}}
\newcommand{\refE}[1]{Equation~(\ref{e:#1})}

\newcommand{\eps}{\ensuremath{\epsilon}}
\newcommand{\Cal}[1]{\ensuremath{\mathcal{#1}}}
\newcommand{\al}{\ensuremath{\alpha}}

\newcommand{\ol}{\overline}

\begin{document}

\raggedbottom
% \reversemarginpar

\numberwithin{equation}{section}

%-------------------- "Style" Commands --------------------
%
%
\newcommand{\MarginNote}[1]{
    \ifthenelse{\value{DetailLevel}=\value{dtlMarginComments} \or
            \value{DetailLevel}>\value{dtlMarginComments}} {
        \marginpar{
            \begin{flushleft}
                \footnotesize #1
            \end{flushleft}
            }
        }
        {}
    }

%
% Define a note-to-self, which remains in the TeX document, but
% produces no output. It will disrupt spacing, unfortunately,
% in the same way that MarginNote will.
%
\newcommand{\NoteToSelf}[1]{
    }

%
% Obsolete material, which should eventually simply be deleted.
%
\newcommand{\Obsolete}[1]{
    }

%
% Tentative material, which might eventually lead somewhere..
%
\newcommand{\Tentative}[1]{
    }

\newcommand{\Detail}[1]{
    \MarginNote{Detail}
    \skipline
    \hspace{+0.25in}\fbox{\parbox{4.25in}{\small #1}}
    \skipline
    }

\newcommand{\Todo}[1]{
    \skipline \noindent \textbf{TODO}:
    #1
    \skipline
    }

\newcommand{\Comment}[1] {
    \skipline
    \hspace{+0.25in}\fbox{\parbox{4.25in}{\small \textbf{Comment}: #1}}
    \skipline
    }

\newcommand{\Ignore}[1]{}

%
% Commonly used expressions in this document.
%

%--- Integral over time and space
\newcommand{\IntTR}
    {\int_{t_0}^{t_1} \int_{\R^d}}

%--- Integral over all \R as -\iny to \iny
\newcommand{\IntAll}
    {\int_{-\iny}^\iny}

%--- Schwartz functions
\newcommand{\Schwartz}
    {\ensuremath \Cal{S}}

%--- Schwartz functions over R
\newcommand{\SchwartzR}
    {\ensuremath \Schwartz (\R)}

%--- Schwartz functions over R^d
\newcommand{\SchwartzRd}
    {\ensuremath \Schwartz (\R^d)}

%--- Schwartz-dual functions
\newcommand{\SchwartzDual}
    {\ensuremath \Cal{S}'}

%--- Schwartz-dual functions over R
\newcommand{\SchwartzRDual}
    {\ensuremath \Schwartz' (\R)}

%--- Schwartz-dual functions over R^d
\newcommand{\SchwartzRdDual}
    {\ensuremath \Schwartz' (\R^d)}

%--- Strange tilde-norm norm on H^s ---
\newcommand{\HSNorm}[1]
    {\norm{#1}_{H^s(\R^2)}}

%--- Strange tilde-norm norm on H^#2 ---
\newcommand{\HSNormA}[2]
    {\norm{#1}_{H^{#2}(\R^2)}}

%--- Holder, with an umlaut over the "o"
\newcommand{\Holder}
    {H\"{o}lder }

%--- Holder's, with an umlaut over the "o"
\newcommand{\Holders}
    {H\"{o}lder's }

%--- Holder, with an umlaut over the "o"
\newcommand{\Holderian}
    {H\"{o}lderian }

%--- Strange tilde-norm (norm on H^s ---
\newcommand{\HolderRNorm}[1]
    {\widetilde{\Vert}{#1}\Vert_r}

%--- L-infinity norm ---
\newcommand{\LInfNorm}[1]
    {\norm{#1}_{L^\iny(\Omega)}}

%--- L-infinity norm, small size ---
\newcommand{\SmallLInfNorm}[1]
    {\smallnorm{#1}_{L^\iny}}

%--- L-1 norm ---
\newcommand{\LOneNorm}[1]
    {\norm{#1}_{L^1}}

%--- Small L-1 norm ---
\newcommand{\SmallLOneNorm}[1]
    {\smallnorm{#1}_{L^1}}

%--- L-2 norm ---
\newcommand{\LTwoNorm}[1]
    {\norm{#1}_{L^2(\Omega)}}

%--- Small L-2 norm ---
\newcommand{\SmallLTwoNorm}[1]
    {\smallnorm{#1}_{L^2}}

%--- L-p norm ---
\newcommand{\LpNorm}[2]
    {\norm{#1}_{L^{#2}}}

%--- Small L-p norm ---
\newcommand{\SmallLpNorm}[2]
    {\smallnorm{#1}_{L^{#2}}}

%--- l-1 norm ---
\newcommand{\lOneNorm}[1]
    {\norm{#1}_{l^1}}

%--- l-2 norm ---
\newcommand{\lTwoNorm}[1]
    {\norm{#1}_{l^2}}

\newcommand{\MsrNorm}[1]
    {\norm{#1}_{\Cal{M}}}

%--- Fourier transform when argument is too long for \widehat
\newcommand{\FTF}
    {\Cal{F}}

%--- Inverse Fourier transform
\newcommand{\FTR}
    {\Cal{F}^{-1}}

\newcommand{\InvLaplacian}
    {\ensuremath{\widetilde{\Delta}^{-1}}}

\newcommand{\EqDef}
    {\hspace{0.2em}={\hspace{-1.2em}\raisebox{1.2ex}{\scriptsize def}}\hspace{0.2em}}

%
%--------------------- Commands specific to this paper -----------------
%

%
%-------------------- Redefined from preamble ---------------------------
%

\title
    [Vanishing viscosity and vorticity on the boundary]
    {Vanishing viscosity and the accumulation of vorticity
    on the boundary}

%--- Information for first author
\author{James P. Kelliher}
% --- Address of record for the research reported here
\address{Department of Mathematics, Brown University, Box 1917, Providence, RI
         02912}
%--- Current address
\curraddr{Department of Mathematics, Brown University, Box 1917, Providence, RI
          02912}
\email{kelliher@math.brown.edu}
% \thanks{}

%    General info
 % ; Secondary }

\date{} %  and, in revised form, .}

% \dedicatory{}


\begin{abstract}
We say that the vanishing viscosity limit holds in the classical sense if the velocity for a solution to the Navier-Stokes equations converges in the energy norm uniformly in time to the velocity for a solution to the Euler equations. We prove, for a bounded domain in dimension $2$ or higher, that the vanishing viscosity limit holds in the classical sense if and only if a vortex sheet forms on the boundary.
% if and only if certain other types of weaker convergence hold, including (in two dimensions) the formation of a vortex sheet on the boundary.
% The \textit{only if} part of this statement exploits only general properties of the convergence of sequences and does not rely on any underlying fluid mechanical principles.
\end{abstract}

\maketitle

%--- Don't include this for arXiv submission, as it gives a timestamp itself
\IfarXivElse{}{
             Compiled on \textit{\textbf{\today}}
             }

%--- Give date and time and detail level UNLESS for submission
\DetailMarginNote{
    \begin{small}
        \begin{flushright}
            Compiled on \textit{\textbf{\today}}

            \ExplainDetailLevel

        \end{flushright}
    \end{small}
    }

%-----------table of contents----------------------
% \tableofcontents

\noindent As is well known, for radially symmetric initial vorticity in a disk, the velocity for a solution to the Navier-Stokes equations converges in the energy norm uniformly in time to the velocity for a solution to the Euler equations. It was shown recently in \cite{FLMT2008} that for such initial data in a disk, it also happens that a vortex sheet forms on the boundary as the viscosity vanishes. By a vortex sheet, we mean a velocity field whose vorticity, as a finite Borel signed measure (an element of the dual space of $C(\ol{\Omega})$), is supported along a curve---the boundary, in this case.

It turns out that this phenomenon is in a sense more universal: either both types of limits hold or neither holds for an arbitrary bounded domain in $\R^d$, $d \ge 2$, with $C^2$-boundary and with no particular assumption on the symmetry of the initial data. More precisely, the vanishing viscosity limit in the classical sense (condition $(B)$ of \refS{VVNoSlip}) holds if and only if a vortex sheet of a particular type forms on the boundary (conditions $(E)$ and $(E_2)$ of \refS{VVNoSlip}). Now, however, the vortex sheet has vorticity belonging to the dual space of $H^1(\Omega)$ rather than $C(\ol{\Omega})$. We show this in \refT{MainResult} for no-slip boundary conditions and in \refT{MainResultb} for characteristic boundary conditions on the velocity.

%
%----------------------------------------------
%
\section{Introduction}\label{S:Introduction}

\noindent Let $\Omega$ be a bounded domain in $\R^d$, $d \ge 2$, with $C^2$-boundary $\Gamma$, and let $\mathbf{n}$ be the outward unit normal vector to $\Gamma$. A classical solution
$(\ol{u}, \ol{p})$ to the Euler equations satisfies
\begin{align*}
    \begin{matrix}
        (EE) & \left\{
            \begin{array}{l}
                \prt_t \ol{u} + \ol{u} \cdot \grad \ol{u} + \grad \ol{p} =
                \ol{f}
                        \text{ and }
                    \dv \ol{u} = 0 \text{ on } [0, T] \times \Omega, \\
                \ol{u}\cdot \mathbf{n} = 0 \text { on } [0, T] \times \Gamma
                \text{ and } \ol{u} = \ol{u}^0 \text{ on } \set{0} \times
                \Omega.
            \end{array}
            \right.
    \end{matrix}
\end{align*}
These equations describe the motion of an incompressible fluid of constant density and zero viscosity. The initial velocity $\ol{u}^0$ must at least lie in
\begin{align*}
    H &= \set{u \in (L^2(\Omega))^d: \dv u = 0 \text{ in } \Omega, \,
                u \cdot \mathbf{n} = 0 \text{ on } \Gamma}
\end{align*}
endowed with the $L^2$-norm, which, along with
\begin{align*}                
    V &= \set{u \in (H^1(\Omega))^d: \dv u = 0 \text{ in } \Omega, \,
                u = 0 \text{ on } \Gamma}
\end{align*}
endowed with the $H^1$-norm, are the classical spaces of fluid mechanics.

We assume that $\ol{u}^0$ is in $C^{k + \eps}(\Omega) \cap H$, $\eps > 0$, where $k =
1$ for two dimensions and $k = 2$ for 3 and higher dimensions, and that
$\ol{f}$ is in $C^1([0, t] \times \Omega)$ for all $t > 0$. Then as shown in
\cite{Koch2002} (Theorem 1 and the remarks on p. 508-509), there is some $T
> 0$ for which there exists a unique solution,
\begin{align}\label{e:ubarSmoothness}
    \ol{u}
        \text{ in } C^1([0, T]; C^{k + \eps}(\Omega)),
\end{align}
to ($EE$). In two dimensions, $T$ can be arbitrarily large, though it is only known that
some nonzero $T$ exists in three and higher dimensions.

The Navier-Stokes equations describe the motion of an incompressible fluid of
constant density and positive viscosity $\nu$. A classical solution to the
Navier-Stokes equations with no-slip boundary conditions can be defined in analogy to ($EE$) by
\begin{align*}
    \begin{matrix}
        (NS) & \left\{
            \begin{array}{l}
                \prt_t u + u \cdot \grad u + \grad p = \nu \Delta u + f
                        \text{ and }
                    \dv u = 0 \text{ on } [0, T] \times \Omega, \\
                u = 0 \text { on } [0, T] \times \Gamma
                \text{ and } u = u_\nu^0 \text{ on } \set{0} \times \Omega,
            \end{array}
            \right.
    \end{matrix}
\end{align*}
where $u^0_\nu$ is in $H$ and $f$ is in $L^1([0, T]; L^2(\Omega))$. We will work, however, with weak solutions to the Navier-Stokes equations. (See, for instance, Chapter III of \cite{T2001}.) Such weak solutions lie in $L^\iny([0, T]; H) \cap L^2([0, T]; V)$.

% In this section we introduced the facts we will need regarding the Navier-Stokes and Euler equations.
In \refS{VVNoSlip} we prove various equivalent forms of the vanishing viscosity limit for no-slip boundary conditions, including the formation of a vortex sheet on the boundary, and remark briefly on their derivation in \refS{Remarks}. In \refS{VVDirichletBCs} we extend the results of Section 2 to characteristic boundary conditions. We discuss, in \refS{FLMT}, our results in relation to those in \cite{FLMT2008} on vortex sheet formation for a disk. Finally, in \refS{Appendix}, we include some technical lemmas employed in \refS{VVNoSlip}.

%
% Section
%
\section{Equivalent forms of the vanishing viscosity limit}\label{S:VVNoSlip}

\noindent Let $\ol{u}$ be a classical solution to ($EE$) in $\Omega$ and $u = u^\nu$ be a weak solution to ($NS$) in $\Omega$ as in \refS{Introduction}, and assume that $u^0_\nu \to \ol{u}^0$ in $H$ and $f \to \ol{f}$ in $L^1([0, T]; L^2(\Omega))$ as $\nu \to 0$.

Let $\gamma_\mathbf{n}$ be the boundary trace operator for the normal component of a vector field (see \refL{IBP}). Let $\Cal{M}(\Omega)$ be the space of finite Borel signed measures on $\ol{\Omega}$---$\Cal{M}(\Omega)$ is the dual space of $C(\ol{\Omega})$. Let $\mu$ in $\Cal{M}(\ol{\Omega})$ be the measure supported on $\Gamma$ for which  $\mu\vert_\Gamma$ corresponds to Lebesgue measure on $\Gamma$ (arc length for $d = 2$, area for $d = 3$, etc.). Then $\mu$ is also a member of $H^1(\Omega)'$.

We define the vorticity $\omega(u)$ to be the $d \times d$ antisymmetric matrix
\begin{align}\label{e:VorticityRd}
    \omega(u) = \frac{1}{2}\brac{\grad u - (\grad u)^T}.
\end{align}
When working specifically in two dimensions, we can alternately define the vorticity as the scalar curl of $u$: 
\begin{align}\label{e:VorticityR2}
	\omega(u) = \prt_1 u^2 - \prt_2 u^1.
\end{align}

Letting $\omega = \omega(u)$ and $\ol{\omega} = \omega(\ol{u})$, we define the following conditions:
\begin{align*}
	(A) & \qquad u \to \ol{u} \text{ weakly in } H
					\text{ uniformly on } [0, T], \\
	(A') & \qquad u \to \ol{u} \text{ weakly in } (L^2(\Omega))^d
					\text{ uniformly on } [0, T], \\
	(B) & \qquad u \to \ol{u} \text{ in } L^\iny([0, T]; H), \\
	(C) & \qquad \grad u \to \grad \ol{u} - \innp{\gamma_\mathbf{n} \cdot, \ol{u} \mu} 
				\text{ in } ((H^1(\Omega))^{d \times d})'
					   \text{ uniformly on } [0, T], \\
	(D) & \qquad \grad u \to \grad \ol{u} \text{ in } (H^{-1}(\Omega))^{d \times d}
					   \text{ uniformly on } [0, T], \\
	(E) & \qquad \omega \to \ol{\omega}
					- \frac{1}{2} \innp{\gamma_\mathbf{n} (\cdot - \cdot^T),
								\ol{u} \mu}
					\text{ in } 
	 				((H^1(\Omega))^{d \times d})'
	 				   \text{ uniformly on } [0, T].
%	(F) & \qquad \omega(u) \to \omega(\ol{u}) \text{ in } 
%	 				(H_0^1(\Omega))^{d \times d}
%	 				   \text{ uniformly on } [0, T].
\end{align*}

In conditions $(C)$, $(D)$, and $(E)$, the convergence is in the weak$^*$ topology of the given spaces. In $(C)$ and $(E)$, $((H^1(\Omega))^{d \times d})'$ is the dual space of $(H^1(\Omega))^{d \times d}$; in $(D)$, $(H^{-1}(\Omega))^{d \times d}$ is the dual space of $(H^1_0(\Omega))^{d \times d}$. Thus, condition ($C$) means that
\begin{align*}
	(\grad u(t), M)
		\to (\grad \ol{u}(t), M) - \int_{\Gamma} (M \cdot \mathbf{n}) \cdot \ol{u}(t)
			\text{ in } L^\iny([0, T])
\end{align*}
for any $M$ in $(H^1(\Omega))^{d \times d}$, condition ($D$) means that
\begin{align*}
	(\grad u(t), M)
		\to (\grad \ol{u}(t), M)
			\text{ in } L^\iny([0, T])
\end{align*}
for any $M$ in $(H^1_0(\Omega))^{d \times d}$, and condition $(E)$ means that
\begin{align*}
	(\omega(t), M)
		\to (\ol{\omega}(t), M) - \frac{1}{2}\int_{\Gamma}
			((M - M^T) \cdot \mathbf{n}) \cdot \ol{u}(t)
				\text{ in } L^\iny([0, T])
\end{align*}
for any $M$ in $(H^1(\Omega))^{d \times d}$.

In two dimensions, defining the vorticity as in \refE{VorticityR2}, we also define the following two conditions:
\begin{align*}
	(E_2) & \qquad \omega \to \ol{\omega} - (\ol{u} \cdot \BoldTau) \mu 
				\text{ in } (H^1(\Omega))'
					   \text{ uniformly on } [0, T], \\
	(F_2) & \qquad \omega \to \ol{\omega} \text{ in } H^{-1}(\Omega)
					   \text{ uniformly on } [0, T].
\end{align*}
Here, $\BoldTau$ is the unit tangent vector on $\Gamma$ that is obtained by rotating the outward unit normal vector $\mathbf{n}$ counterclockwise by $90$ degrees.

Condition ($E_2$) means that
\begin{align*}
	(\omega(t), f)
		\to (\ol{\omega}(t), f) - \int_{\Gamma} (\ol{u}(t) \cdot \BoldTau) f
			\text{ in } L^\iny([0, T])
\end{align*}
for any $f$ in $H^1(\Omega)$, while condition ($F_2$) means that
\begin{align*}
	(\omega(t), f)
		\to (\ol{\omega}(t), f)
			\text{ in } L^\iny([0, T])
\end{align*}
for any $f$ in $H_0^1(\Omega)$.

\begin{theorem}\label{T:MainResult}
	Conditions ($A$), ($A'$), ($B$), ($C$), ($D$), and ($E$) are equivalent.
	In two dimensions, conditions ($E_2$) and ($F_2$) are equivalent to the other conditions.
\end{theorem}
\begin{proof}

$\mathbf{(A) \iff (A')}$: Let $v$ be in $(L^2(\Omega))^d$. By \refL{LHW}, $v = w + \grad p$, where $w$ is in $H$ and $p$ is in $H^1(\Omega)$. Then assuming $(A)$ holds,
\begin{align*}
	(u(t), v)
		&
		= (u(t), w)
		\to (\ol{u}(t), w)
		= (\ol{u}(t), v)
		\end{align*}
uniformly over $t$ in $[0, T]$, so $(A')$ holds. The converse is immediate.

	\medskip

	\noindent $\mathbf{(A) \iff (B)}$: The forward implication is proved in Theorem 1
	of \cite{Kato1983}. The backward implication is immediate.
	
	\medskip
	
	\noindent $\mathbf{(A') \implies (C)}$: Assume that ($A'$) holds and let $M$ be in
	$(H^1(\Omega))^{d \times d}$. Then
	\begin{align*}
		(\grad u(t), M&)
			= - (u(t), \dv M) 
			 \to -(\ol{u}(t), \dv M)
					\text{ in } L^\iny([0, T]).
	\end{align*}
	But,
		\begin{align*}
			-(\ol{u}(t), \dv M)
					= (\grad \ol{u}(t), M)
									- \int_\Gamma (M \cdot \mathbf{n}) \cdot \ol{u},
	\end{align*}
	giving ($C$).
	
	\medskip
	
\noindent $\mathbf{(C) \implies (D)}$: This follows simply because $H^1_0(\Omega) \subseteq H^1(\Omega)$.

	\medskip
	
	% \noindent $\mathbf{(D) \iff (D')}$: That $(D) \implies (D')$ follows directly
	% from \refE{VorticityDef}.

	\medskip
	
	\noindent $\mathbf{(D) \implies (A)}$: Assume ($D$) holds, and let $v$ be
	in $H$. Then $v = \dv M$ for some $M$ in $(H^1_0(\Omega))^{d \times d}$ by
	\refC{Cdiv}, so
	\begin{align*}
		(u(t), &v)
			= (u(t), \dv M)
			=  -(\grad u(t), M) + \int_\Gamma (M \cdot \mathbf{n}) \cdot u(t) \\
			&= -(\grad u(t), M)
			\to -(\grad \ol{u}(t), M),
	\end{align*}
	uniformly over $[0, T]$.
	But,
	\begin{align*}
		- (\grad \ol{u}(t), M)
			= (\ol{u}(t), \dv M) - \int_\Gamma (M \cdot \mathbf{n}) \cdot \ol{u}(t)
				= (\ol{u}(t), v),
	\end{align*}
	from which ($A$) follows.
	
	\medskip
	
	Now assume that $d = 2$.
	
	\medskip

	\noindent $\mathbf{(A') \implies (E_2)}$: Assume that ($A'$) holds and let $f$ be in
	$H^1(\Omega)$. Then
	\begin{align*}
		(\omega(t), f&)
			= - (\dv u^\perp(t), f)
			= (u^\perp(t), \grad f)
			= - (u(t), \grad^\perp f) \\
			&\to -(\ol{u}(t), \grad^\perp f)
					\text{ in } L^\iny([0, T])
	\end{align*}
	where $u^\perp = -\innp{u^2, u^1}$ and we used the identity $\omega(u) = - \dv u^\perp$.
	But,
		\begin{align*}
			-(\ol{u}(t), &\grad^\perp f)
					= (\ol{u}^\perp(t), \grad f) 
					= - (\dv \ol{u}^\perp(t), f) 
							+ \int_\Gamma (\ol{u}(t)^\perp \cdot \mathbf{n}) f \\
					&= - (\dv \ol{u}^\perp(t), f) 
							- \int_\Gamma (\ol{u}(t) \cdot \BoldTau) f 
					= (\ol{\omega}(t), f) 
							- \int_\Gamma (\ol{u}(t) \cdot \BoldTau) f,
	\end{align*}
	giving ($E_2$).
	
	\medskip
	
	\noindent $\mathbf{(E_2) \implies (F_2)}$: Follows for the same reason that
	$(C) \implies (D)$.
	
	\medskip
	
	\noindent $\mathbf{(F_2) \implies (A)}$: Assume ($F_2$) holds, and let $v$ be
	in $H$. Then $v = \grad^\perp f$ for some $f$ in $H^1_0(\Omega)$ ($f$ is called
	the stream function for $v$), and
	\begin{align*}
		(u(t), &v)
			= (u(t), \grad^\perp f)
			= - (u^\perp(t), \grad f)
			= (\dv u^\perp(t), f) \\
			&= - (\omega(t), f)
			\to - (\ol{\omega}(t), f) \text{ in } L^\iny([0, T]).
	\end{align*}
	But,
	\begin{align*}
		- (\ol{\omega}(t), &f)
			= (\dv \ol{u}^\perp(t), f)
			= - (\ol{u}^\perp(t), \grad f)
			= (\ol{u}(t), \grad^\perp f) \\
			&= (\ol{u}(t), v),
	\end{align*}
	which shows that ($A$) holds.
	
What we have shown so far is that ($A$), ($A'$), ($B$), ($C$), and ($D$) are equivalent, as are $(E_2)$ and $(F_2)$ in two dimensions. It remains to show that $(E)$ is equivalent to these conditions as well. We do this by establishing the implications $(C) \implies (E) \implies (A)$.

\medskip

\noindent $\mathbf{(C) \implies (E)}$: Follows directly from \refE{VorticityRd}.

\medskip

\noindent $\mathbf{(E) \implies (A)}$: Let $v$ be in $H$ and let $x$ be the vector field in $(H^2(\Omega) \cap H_0^1(\Omega))^d$ solving $\Delta x = v$ on $\Omega$ ($x$ exists and is unique by standard elliptic theory). Then, utilizing \refL{omegavsgrad} twice (and suppressing the explicit dependence of $u$ and $\ol{u}$ on $t$),
\begin{align}\label{e:EImpliesAEquality}
	\begin{split}
	(u, v)
		&= (u, \Delta x)
		= - (\grad u, \grad x) + \int_\Gamma (\grad x \cdot \mathbf{n}) \cdot u
		= - (\grad u, \grad x) \\
		&= -2 (\omega(u), \omega(x)) - \int_\Gamma (\grad u x) \cdot \mathbf{n}
		= -2 (\omega(u), \omega(x)) \\
		&\to -2(\omega(\ol{u}), \omega(x))
				+ 2 \frac{1}{2} \int_\Gamma((\omega(x) - \omega(x)^T) \cdot \mathbf{n})
							\cdot \ol{u} \\
		&= -2(\omega(\ol{u}), \omega(x))
				+ 2 \int_\Gamma(\omega(x) \cdot \mathbf{n})
							\cdot \ol{u} \\
		&= -(\grad \ol{u}, \grad x)
				+ \int_\Gamma (\grad \ol{u} x) \cdot \mathbf{n}
				+ 2 \int_\Gamma(\omega(x) \cdot \mathbf{n})
							\cdot \ol{u} \\
		&= -(\grad \ol{u}, \grad x)
				+ 2 \int_\Gamma(\omega(x) \cdot \mathbf{n})
							\cdot \ol{u} \\
		&= (\ol{u}, \Delta x)
				- \int_\Gamma (\grad x \cdot \mathbf{n}) \cdot \ol{u}
				+ 2 \int_\Gamma(\omega(x) \cdot \mathbf{n})
							\cdot \ol{u} \\
		&= (\ol{u}, v)
				- \int_\Gamma ((\grad x)^T \cdot \mathbf{n}) \cdot \ol{u}.
	\end{split}
\end{align}
Thus, $(E) \implies (A)$ if and only if
\begin{align}\label{e:EImpliesACondition}
	\int_\Gamma ((\grad x)^T \cdot \mathbf{n}) \cdot \ol{u}
		= 0.
\end{align}

But, $(\dv (\grad x)^T)^j = \prt_j \prt_i x^j = \prt_i \dv x$ or $\dv (\grad x)^T = \grad \dv x$. Similarly, $\dv (\grad \ol{u})^T = \grad \dv \ol{u} = 0$. It follows that
\begin{align*}
	\int_\Gamma ((&\grad x)^T \cdot \mathbf{n}) \cdot \ol{u}
		 = ((\grad x)^T, \grad \ol{u}) + (\grad \dv x, \ol{u}) \\
		&= (\grad x, (\grad \ol{u})^T)
			- (\dv x, \dv \ol{u}) + \int_\Gamma \ol{u} \cdot \mathbf{n} \dv x
		= ((\grad \ol{u})^T, \grad x) \\
		&= - (\dv (\grad \ol{u})^T, x)
					+ \int_\Gamma ((\grad \ol{u})^T \cdot \mathbf{n}) x
		= 0.
\end{align*}

\end{proof}

%
% Section
%
\section{Remarks}\label{S:Remarks}

\noindent

The equivalent conditions of \refT{MainResult} complement those of \cite{Kato1983}, \cite{TW1998}, \cite{W2001}, and \cite{K2006Kato}.

It is only in the proof of $(A) \implies (B)$---in which we quote a result of Kato's in \cite{Kato1983}---where the requirement that $\ol{u}$ be a classical solution to the Euler equations and that $f \to \ol{f}$ in $L^1([0, T]; L^2)$ is used; in fact, it is the only place where the fact that $u$ and $\ol{u}$ are solutions to the Navier-Stokes and Euler equations, respectively, appear in the proof at all. That is, assuming only that $u$ is a vector field parameterized by $\nu$ that lies in $L^\iny([0, T]; H) \cap L^2([0, T]; V)$ and that $\ol{u}$ is a vector field lying in $L^\iny([0, T]; H \cap H^1(\Omega))$, all of the implications in the proof of \refT{MainResult} remain valid except for $(A) \implies (B)$.

The proof of $(A) \implies (B)$ in \cite{Kato1983} consists, using our terminology, of proving $(B) \implies (A) \implies (i) \implies (ii) \implies (B)$, where $(i)$ and $(ii)$ are the conditions,
\begin{align*}
	(i) & \qquad \nu \int_0^T \norm{\grad u}_{L^2(\Omega)}^2 \to 0, \\
	(ii) & \qquad \nu \int_0^T \norm{\grad u}_{L^2(\Gamma_{c \nu})}^2 \to 0,
\end{align*}
with $\Gamma_{c \nu}$ a boundary layer of width proportional to $\nu$.

If we weaken the regularity of $\Gamma$ from $C^2$ to only locally Lipschitz, then the proof of $(E) \implies (A)$ fails because we would only have $x$ in $(H^1_0(\Omega))^d$. Kato's proof that $(A) \implies (B)$ also requires a $C^2$ boundary.

In two dimensions, we need only have convergence of the vorticity away from the boundary---condition $(F_2)$---to insure that the vanishing viscosity limit holds. In particular, it follows that formation of a vortex sheet on the boundary of a type other than that given in $(E_2)$ is inconsistent with $u$ being a solution to ($NS$). In higher dimensions it is an open problem whether the analogous statement is true; that is, whether $(F) \implies (A)$, where $(F)$ is the condition,
\begin{align*}
	(F) & \qquad \omega \to \ol{\omega} \text{ in } 
	 				(H^{-1}(\Omega))^{d \times d}
	 				   \text{ uniformly on } [0, T].
\end{align*}
The remarks that follow attempt to give some insight into the nature of this problem. 

One approach to proving that $(F) \implies (A)$ is to prove that $(F) \implies (D)$, since we have $(D) \implies (A)$.
So suppose that $(F$) holds, and let $M$ be in $(H_0^1(\Omega))^{d \times d}$. For any vector field $v$,
\begin{align*}
	(\grad v, M)
		&= (\grad v - (\grad v)^T, M) + ((\grad v)^T, M)
		= 2 (\omega(v), M) + (\grad v, M^T).
\end{align*}
Thus,
\begin{align*}
	(\grad u, M - M^T)
		= 2 (\omega(u), M)
		\to 2 (\omega(\ol{u}), M)
		= (\grad \ol{u}, M - M^T).
\end{align*}
If $M$ is antisymmetric then $M - M^T = 2M$ and we conclude that $(D)$ holds for antisymmetric matrix fields in $(H_0^1(\Omega))^{d \times d}$. But if $M$ in $(H_0^1(\Omega))^{d \times d}$ is symmetric,
\begin{align*}
	2 (\omega(u), M)
		= (\grad u, M) - ((\grad u)^T, M)
		= (\grad u, M - M^T)
		= 0,
\end{align*}
so $(\omega(u), M) = (\omega(\ol{u}), M) = 0$, and we can conclude nothing from this approach.

But some use can still be made of this observation. Let $v$ be any element of $H$. Then from \refC{Cdiv}, for some $M$ in $(H_0^1(\Omega))^{d \times d}$,
\begin{align*}
	(u, v)
		&= (u, \dv M)
		= -(\grad u, M).
\end{align*}
Now, if we could insure that $M$ can be chosen to be antisymmetric, then if $(F)$ holds for $M$ so does $(D)$, as we just showed, and
\begin{align*}
	-(\grad u, M) \to -(\grad \ol{u}, M)
	= (\ol{u}, v),
\end{align*}
and $(A)$ would follow.

In two dimensions, we can choose
\begin{align}\label{e:MR2}
	M =
		\begin{pmatrix}
			0 & -f \\
			f & 0
		\end{pmatrix},
\end{align}
where $f$ is the stream function for $v$ as in the proof of $(F_2) \implies (A)$, which gives a slight variation on the proof of that same implication.

In three dimensions, however, it is not possible to find such an $M$. To see this, suppose that $M$ in $(H_0^1(\Omega))^{d \times d}$ is antisymmetric. Then can write $M$ as
\begin{align*}
	M
		&= \begin{pmatrix}
				0 & a & b \\
				-a & 0 & c \\
				-b & -c & 0
		\end{pmatrix}
\end{align*}
with $a = b = c = 0$ on $\Gamma$. In this form the condition $\dv u = \dv \dv M = 0$ is automatically satisfied, and letting $\tilde{\omega}$ be the vector $\innp{c, -b, a}$, we see that
\begin{align}\label{e:Momega}
	u
		= \dv M
		= \curl \tilde{\omega},
\end{align}
where $\curl$ is the usual three-dimensional operator. But $\curl$ maps $H \cap C^\iny(\Omega)$ bijectively onto itself when $\Gamma$ is $C^\iny$ (see, for instance, \cite{CDG2000}), so in general we only have $\tilde{\omega} \cdot \mathbf{n} = 0$ on $\Gamma$. That is, the condition that $M$ be antisymmetric is not compatible with the condition that it vanish on $\Gamma$.

Finally, let
\begin{align*}
	E(\Omega) = \set{u \in (L^2(\Omega))^d \colon \dv u \in L^2(\Omega),
				u \cdot \mathbf{n} = 0 \text{ on } \Gamma},
\end{align*}
with $\norm{u}_{E(\Omega)} = \norm{u}_{L^2(\Omega)} + \norm{\dv u}_{L^2(\Omega)}$.
It is easy to see from the proofs of $(A') \implies (C)$ and $(D) \implies (A)$ that condition $(D)$ can be weakened from convergence in the dual space of $(H_0^1(\Omega))^{d \times d}$ to convergence in the dual space of $(E(\Omega))^d$. (This is advantageous as a sufficient condition, though not as a necessary one.)

Returning to \refE{Momega}, the condition $\tilde{\omega} \cdot \mathbf{n} = 0$ on $\Gamma$ does not translate to $M \cdot \mathbf{n} = 0$ on $\Gamma$. Hence, $M$ does not lie in $(E(\Omega))^d$ so we cannot use this weakening of condition $(D)$ to conclude that $(A)$ holds.

%
% Section
%
\section{Characteristic boundary conditions on the velocity}\label{S:VVDirichletBCs}

\noindent We modify ($NS$) by allowing the velocity on the boundary to be equal to a nonzero time-varying vector field $b$, which is required, however, to be tangential to the boundary. This gives
\begin{align*}
    \begin{matrix}
        (NS_b) & \left\{
            \begin{array}{l}
                \prt_t u + u \cdot \grad u + \grad p = \nu \Delta u + f
                        \text{ and }
                    \dv u = 0 \text{ on } [0, T] \times \Omega, \\
                u = b \text { on } [0, T] \times \Gamma
                \text{ and } u = u_\nu^0 \text{ on } \set{0} \times \Omega,
            \end{array}
            \right.
    \end{matrix}
\end{align*}
where, as before, $u^0_\nu$ is in $H$.

We\MarginNote{What is the best reference giving the existence results for nonzero $b$?}require sufficient regularity on $b$ so that $(NS_b)$ is well-posed. For simplicity, we will assume that $b \cdot \mathbf{n} = 0$ on $\Gamma$ with
\begin{align}\label{e:bRegularity}
	b \in L^\iny([0, T]; H^{3/2}(\Gamma)), \;
	\prt_t b \in L^2([0, T]; H^{-1/2}(\Gamma))
\end{align}
so that $b$ lifts (extends) to a vector field, which we also call $b$, with
\begin{align*}
	b \in L^\iny([0, T]; H \cap H^2(\Omega)), \;
	\prt_t b \in L^2([0, T]; L^2(\Omega)).
\end{align*}
The assumption on $b$ in \refE{bRegularity} is not the weakest possible, but the assumption on $\prt_t b$ can probably not be weakened.

We can then use the equation that corresponds to $u - b$ (which lies in $L^\iny([0, T]; H) \cap L^2([0, T]; V)$ for classical solutions) to define a weak solution to ($NS_b$). (This is essentially what is done in Section 4 of \cite{W2001}.)

With such solutions to $(NS_b)$ in place of those for ($NS$)---but without changing the formulation of ($EE$)---we define the condition,
\begin{align*}
	(C^b) & \qquad \grad u \to \grad \ol{u} - \innp{\gamma_\mathbf{n} \cdot,
				(\ol{u} - b) \mu}  \text{ in } ((H^1(\Omega))^{d \times d})'
	 				   \text{ uniformly on } [0, T], \\
	(E^b) & \qquad \omega \to \ol{\omega}
					- \frac{1}{2} \innp{\gamma_\mathbf{n} (\cdot - \cdot^T),
								(\ol{u} - b) \mu}
					\text{ in } 
	 				((H^1(\Omega))^{d \times d})' \\
		& \qquad
	 				   \text{uniformly on } [0, T],
\end{align*}
and in two dimensions, the condition,
\begin{align*}
	(E_2^b) & \qquad \omega \to \ol{\omega} - ((\ol{u} - b) \cdot \BoldTau) \mu 
				\text{ in } (H^1(\Omega))'
					   \text{ uniformly on } [0, T].
\end{align*}

\refT{MainResult} then becomes:
\begin{theorem}\label{T:MainResultb}
Let $u$ be a solution to ($NS_b$) and $\ol{u}$ be a solution to ($EE$). Conditions ($A$), ($A'$), ($B$), ($C^b$), ($D$), and ($E^b$) are equivalent. In two dimensions, conditions ($E_2^b$) and ($F_2$) are equivalent to the other conditions.
\end{theorem}
\begin{proof}

The proof of $(A') \implies (C^b)$ is identical to the proof of $(A') \implies (C)$ except that a boundary term is included in the first step: this term leads to the ``$- b$'' in condition ($C^b$). A similar comment applies to the proof of $(A') \implies (E_2^b)$ and $(C) \implies (E^b)$. The proof of $(C) \implies (D)$ and $(E_2) \implies (F_2)$ are unaffected by the presence of the vector field $b$, as are all other implications except for $(A) \implies (B)$ and $(E^b) \implies (A)$.

\medskip

\noindent $\mathbf{(A) \implies (B)}$: In \cite{W2001} it is shown, using our terminology, that $(B) \implies (i) \implies (ii') \implies (B)$, where ($ii'$) is the condition
\begin{align*}
	(ii') & \qquad \nu \int_0^T \norm{\grad_{\BoldTau} u_{\BoldTau}}_{L^2(\Gamma_{\delta(\nu)})}^2 \to 0
				\text{ or } 
	     \nu \int_0^T \norm{\grad_{\BoldTau} u_\mathbf{n}}_{L^2(\Gamma_{\delta(\nu)})}^2 \to 0.
\end{align*}
Here, $\grad_{\BoldTau}$ is the gradient only in the direction tangential to the boundary, $u_{\BoldTau}$ and $u_\mathbf{n}$ are the components of the velocity tangential and normal to the boundary, respectively, and $\Gamma_{\delta(\nu)}$ is a boundary layer whose width $\delta(\nu)$ is of arbitrary order larger than $\nu$.

But, $(B) \implies (A)$ is immediate, and $(A) \implies (i)$ follows from combining the argument on pages 232 through 233 of \cite{W2001} with the proof of the implication $(A) \implies (i)$ on page 90 of \cite{Kato1983} (in Kato's terminology, this is (ii) implies (iii)). This gives $(A) \implies (B)$.

\medskip

\noindent $\mathbf{(E^b) \implies (A)}$: Adapting the argument of $(E) \implies (A)$, we see that the first boundary integral in \refE{EImpliesAEquality} does not vanish, since now $u = b$ on $\Gamma$. Also, $\ol{u}$ becomes $\ol{u} - b$ in the boundary integrals involving $\omega(x)$. This leads to
\begin{align*}
	(u, v)
		&= (\ol{u}, \Delta x)
				+ \int_\Gamma (\grad x \cdot \mathbf{n}) \cdot b
				- \int_\Gamma (\grad x \cdot \mathbf{n}) \cdot \ol{u}
				+ 2 \int_\Gamma(\omega(x) \cdot \mathbf{n})
							\cdot (\ol{u} - b) \\
		&= (\ol{u}, v)
				- \int_\Gamma ((\grad x)^T \cdot \mathbf{n}) \cdot (\ol{u} - b).
\end{align*}
The last boundary integral vanishes for the same reason that \refE{EImpliesACondition} holds, $\ol{u} - b$ being in $H$, completing the proof.
\end{proof}

%
% Section
%
\section{Radially symmetric initial vorticity in a disk}\label{S:FLMT}

\noindent We assume, in this section only, that $\Omega$ is the unit disk $D$ and that the initial vorticity $\ol{\omega}^0$ is radially symmetric. In this case, the solution to ($EE$) is stationary: $\ol{\omega}(t) = \ol{\omega}^0$ for all time $t$.

The vanishing viscosity limit in the classical sense of condition $(B)$ holds in this setting under fairly general circumstances. Under the assumptions of \refS{Introduction} on the regularity of the initial velocity, and assuming that $b = 0$, this convergence is implicit in \cite{Kato1983} at least for zero forcing (see \cite{K2006Disk}), but was first explicitly\MarginNote{Did Matsui consider nonzero forcing?}proved by Matsui (see \cite{Matsui1994}). For nonzero $b$ having the regularity assumed in \refE{bRegularity}, the convergence is a simple consequence of the condition $(ii')$ in the proof of \refT{MainResultb} as observed by Wang in \cite{W2001}. For substantially lower regularity on $\ol{u}^0$ and on $b$ than we assume, the convergence is established in \cite{FLMT2008}.

More precisely, the authors of \cite{FLMT2008} assume that $u^0 = \ol{u}^0$ and $f = \ol{f} = 0$ (which are not significant limitations, since one can handle $u^0 \to \ol{u}^0$ in $H$ by using the triangle inequality, and nonzero forcing presents no real difficulties), with
\begin{align*}
	\ol{u}^0 \in R^1(D)
		&= \set{v \in (L^2(D))^2 \colon
				v(x) = s(\abs{x}) x^\perp \text{ for some } s, \, \omega(v) \in L^1(D)} \\
		&= \set{v \in H \colon \omega(v) \in L^1(D), \, \omega(v)
				\text{ radially symmetric}}.
\end{align*}
They assume that $b(t, \cdot) = \al(t)$---that is, $b(t)$ is constant on the boundary---and that $\al \in \BV([0, T])$, the space of bounded variation functions. They prove (combining Propositions 9.6 and 9.7 of \cite{FLMT2008}) that
\begin{align}\label{e:FLMTConvergence}
	\omega
		\to \ol{\omega} - (B (2 \pi)^{-1} - b \cdot \BoldTau) \mu
				\text{ in } \Cal{M}(\ol{D})
					   \text{ uniformly on } [0, T],
\end{align}
where
\begin{align*}
	B = \int_D \ol{\omega}^0.
\end{align*}
But, on $\Gamma$, $u^0 \cdot \BoldTau$ is constant, so by Green's theorem,
\begin{align*}
	B = \int_\Gamma \ol{u}^0 \cdot \BoldTau
		= 2 \pi \ol{u}^0(x) \cdot \BoldTau(x)
\end{align*}
for any point $x$ on $\Gamma$, and we see that \refE{FLMTConvergence} is the same as condition $(E_2^b)$, except that the convergence is stronger.

That is, both conditions $(B)$ and $(E_2^b)$ hold for a disk, except that the convergence in $(E_2^b)$ is in $\Cal{M}(\ol{\Omega})$, which is stronger convergence than that of $(E_2^b)$. What we have shown is that either both conditions $(B)$ and $(E_2^b)$ hold or neither condition holds for a given initial velocity in a general bounded domain in the plane---and in the analogous sense, in $\R^d$. It was the question of whether this was, in fact, the case that motivated this paper.

The regularity we assume in \refE{bRegularity} corresponds to $\al$ lying in $H^1([0, T])$, which is considerably stronger than the assumption in \cite{FLMT2008} that $\al$ lie in $\BV([0, T])$. And their assumption on the regularity of $\ol{u}^0$ is far lower than our assumption that $\ol{u}^0$ lies in $C^{1 + \eps}(\Omega) \cap H$. Without the assumption of radial symmetry, however, it seems unlikely that one can weaken our assumptions in \refE{bRegularity} on $b$ to any significant degree, since these assumptions go to the heart of establishing the existence of the corresponding weak solutions of ($NS_b$). Weakening the regularity assumptions on $\ol{u}^0$ would seem equally impossible, since the boundedness of $\grad \ol{u}$ on $[0, T] \times \Omega$ is indispensable in Kato's argument showing that $(A) \implies (B)$.

\Ignore{
The proof of \refE{FLMTConvergence} in \cite{FLMT2008} follows quite a different path than we have described here of establishing $(B)$ then applying \refT{MainResultb}. A step along the way is establishing a type of internal convergence of the vorticity. Specifically, for the case of no-slip boundary conditions ($b = 0$), they show (Proposition 9.5 of \cite{FLMT2008}) that
\begin{align*}
	\norm{\omega - \ol{\omega}}_{L^\iny([0, T]; L^1(\Omega'))} \to 0
		\text{ as } \nu \to 0,
\end{align*}
where $\Omega'$ is any open subset of $\Omega$ that is bounded away from $\Gamma$. This along with a bound they establish on the $L^\iny([0, T]; L^1(\Omega))$-norm of $\omega$ gives \refE{FLMTConvergence} for $b = 0$.

The reason that the $L^\iny([0, T]; L^1(\Omega))$-norm of $\omega$ is difficult to obtain in \refE{FLMTConvergence} is that the regularity of $u^0$ is such that, even with radial symmetry, $\omega^0$ is only in $L^1(\Omega)$. 

An interesting question is whether .....
}

%
% Section
%
\section{Some technical lemmas}\label{S:Appendix}

\noindent In this section we assume only that $\Omega$ is bounded and that $\Gamma$ is locally Lipschitzian, which of course includes the case that $\Gamma$ is $C^2$.

The various integrations by parts that we make are justified by \refL{IBP}, which is Theorem 1.2 p. 7 of \cite{T2001} for locally Lipschitz domains. (Temam states the theorem for $C^2$ boundaries but the proof for locally Lipschitz boundaries is the same, using a trace operator for Lipschitz boundaries in place of that for $C^2$ boundaries: see p. 117-119 of \cite{G1994}, in particular, Theorem 2.1 p. 119.)

\begin{lemma}\label{L:IBP}
Let
\begin{align*}
	E(\Omega)
		= \set{v \in (L^2(\Omega))^d \colon \dv v \in L^2(\Omega)}
\end{align*}
with $\norm{v}_{E(\Omega)} = \norm{v}_{L^2(\Omega)} + \norm{\dv v}_{L^2(\Omega)}$.
There exists an extension of the trace operator $\gamma_\mathbf{n} \colon (C_0^\iny(\ol{\Omega}))^d \to C^\iny(\Gamma)$ defined by $u \mapsto u \cdot \mathbf{n}$ on $\Gamma$ to a continuous linear operator from $E(\Omega)$ onto $H^{-1/2}(\Gamma)$. The kernel of $\gamma_\mathbf{n}$ is the space $E_0(\Omega)$---the completion of $C_0^\iny(\Omega)$ in the $E(\Omega)$ norm. For all $u$ in $E(\Omega)$ and $f$ in $H^1(\Omega)$,
\begin{align}\label{e:IBPEq}
	(u, \grad f) + (\dv u, f) = \int_\Gamma (u \cdot \mathbf{n}) \ol{f}.
\end{align}
\end{lemma}

\begin{lemma}\label{L:uvp}
	Assume that $u$ is in $(\Cal{D}'(\Omega))^d$ with $(u, v)$ = 0 for all $v$ in
	$\Cal{V}$. Then $u = \grad p$ for some $p$ in $\Cal{D}'(\Omega)$. If $u$ is in
	$(L^2(\Omega))^d$ then $p$ is in $H^1(\Omega)$; if $u$ is in $H$ then $p$ is
	in $H^1(\Omega)$ and $\Delta p = 0$.
\end{lemma}
\begin{proof}
For $u$ in $(\Cal{D}'(\Omega))^d$ see Proposition 1.1 p. 10 of \cite{T2001}. For $u$ in $(L^2(\Omega))^d$ the result follows from a combination of Theorem 1.1 p. 107 and Remark 4.1 p. 55 of \cite{G1994} (also see Remark 1.4 p. 11 of \cite{T2001}).
\end{proof}

\begin{lemma}\label{L:LHW}
	For any $u$ in $(L^2(\Omega))^d$ there exists a unique $v$ in $H$ and $p$ in
	$H^1(\Omega)$
	such that $u = v + \grad p$.
\end{lemma}
\begin{proof}
	This follows, for instance, from Theorem 1.1 p. 107 of \cite{G1994}, which holds for
	an arbitrary domain, along with \refL{uvp}.
\end{proof}

\begin{lemma}\label{L:div}
	For any $f$ in $L^2(\Omega)$ and $a$ in $(H^{1/2}(\Gamma))^d$ satisfying 
	the compatibility condition,
	\begin{align*}
		\int_\Omega f = \int_\Gamma a \cdot \mathbf{n}
	\end{align*}
	there exists a (non-unique) solution $v$ in $(H^1(\Omega))^d$ to $\dv v = f$ 
	in $\Omega$, $v = a$ on $\Gamma$.
\end{lemma}
\begin{proof}
	This follows from Lemma 3.2 p. 126-127, Remark 3.3 p. 128-129, and Exercise 
	3.4 p. 131 of \cite{G1994} (and see the comment on p. 67 of \cite{Adams1978}).
\end{proof}

\begin{cor}\label{C:Cdiv}
	For any $v$ in $H$ there exists a matrix-valued function $M$ in
	$(H^1_0(\Omega))^{d \times d}$ such that $v = \dv M$.
\end{cor}
\begin{proof}
	Let $v$ be in $H$ and observe that
	\begin{align*}
		\int_\Omega v^i
			&= \int_\Omega v \cdot \grad x_i
			= - \int_\Omega \dv v \, x_i + \int_\Gamma (v \cdot \mathbf{n}) x_i
			= 0.
	\end{align*}
	Thus, we can apply \refL{div} to each component $v^i$ using $a \equiv 0$ to obtain
	a vector $w^i$ in $H^1_0(\Omega)$ satisfying $\dv w^i = v^i$ on $\Omega$, $w^i = 0$
	on $\Gamma$. Forming a matrix-valued
	function whose rows are $w^1, w^2, \dots, w^d$ gives $M$.
\end{proof}

\begin{lemma}\label{L:omegavsgrad}
	Assume that $u$ is in $(H^2(\Omega))^d$ or $(C^1(\Omega))^d$ with $\dv u = 0$
	and $v$ is in $(H^1(\Omega))^d$. Then
   \begin{align*}
        (\grad u, \grad v)
            = 2 (\omega(u), \omega(v))
            		+ \int_\Gamma (\grad u v)
                    \cdot \mathbf{n}.
    \end{align*}
\end{lemma}
\begin{proof}
Directly from \refE{VorticityRd},
\begin{align}\label{e:VorticityIdentity}
    \begin{split}
        2\omega(u) \cdot \omega(v)
        &= \frac{1}{2} (\grad u - (\grad u)^T) \cdot (\grad v - (\grad v)^T) \\
        &= \grad u \cdot \grad v - (\grad u)^T \cdot \grad v.
    \end{split}
\end{align}

    Since $\dv u = 0$, we have $(\grad u)^T \cdot \grad v = \prt_j v^i \prt_i
    u^j = \prt_j (v^i \prt_i u^j) = \dv (\grad u v)$, so if $u$ and $v$ are
    both in $(C^\iny(\Omega))^d$ with $\dv u = 0$ then
    \begin{align*}
        2 \int_\Omega &\omega(u) \cdot \omega(v)
            = \int_\Omega \grad u \cdot \grad v - (\grad u)^T \cdot \grad v \\
           &= \int_\Omega \grad u \cdot \grad v - \int_\Omega \dv (\grad u v)
            = \int_\Omega \grad u \cdot \grad v - \int_\Gamma (\grad u v)
                    \cdot \mathbf{n}.
    \end{align*}
    The result then follows by the density of $C^\iny(\Omega)$ in $H^1(\Omega)$,
    $H^2(\Omega)$, and $C^1(\Omega)$.
\end{proof}

\section*{Acknowledgements}

\noindent The author would like to thank Helena Nussenzveig Lopes, Milton Lopes Filho, and Robert Pego for valuable comments on early versions of this paper. The author was supported in part by NSF grant DMS-0705586 during the period of this work.

\IfarXivElse{

} % End bibliography if for arXiv.org
{ % Else get bibliography from ref.bib
\bibliography{Refs}
\bibliographystyle{plain}
}

\end{document}